\documentclass[10pt,journal]{IEEEtran}

\usepackage{amsmath}
\usepackage{amsfonts}
\usepackage{amssymb}
\usepackage{amsthm}
\usepackage{tabularx}
\usepackage{mathrsfs} 

\usepackage{url}
\usepackage{hyperref}
\newtheorem{theorem}{Theorem}
\newtheorem{definition}{Definition}
\newtheorem{Remark}{Remark}
\newtheorem{lemma}{Lemma}
\newtheorem{corollary}{Corollary}

\usepackage{stfloats}
\usepackage{float}
\usepackage{graphicx}
\usepackage{cite}
\usepackage{xcolor}
\renewcommand{\vec}[1]{\mathbf{#1}}

\makeatletter
\def\blfootnote{\xdef\@thefnmark{}\@footnotetext}
\makeatother

\begin{document}
	
		\title{The Impact of Side Information on Physical Layer Security under Correlated Fading Channels} 
	\author{Farshad~Rostami~Ghadi\IEEEmembership{}, F.~Javier~L{\'o}pez-Mart{\'i}nez, Wei-Ping~Zhu, and Jean-Marie~Gorce  \IEEEmembership{}\IEEEmembership{}}
	\maketitle
	\begin{abstract}
		In this paper, we investigate the impact of side information (SI) on the performance of physical layer security (PLS) under correlated fading channels. By considering non-causally known SI at the transmitter and exploiting the copula technique to describe the fading correlation, we \textcolor{blue}{derive} closed-from expressions for the average secrecy capacity (ASC) and secrecy outage probability (SOP) under
		positive/negative dependence conditions. We indicate that considering such knowledge at the transmitter is beneficial for system performance and ensures reliable communication with higher rates, as it improves the SOP and brings higher values of the ASC.
	\end{abstract}
	\maketitle
	\blfootnote{\noindent Manuscript received December 21, 2021; revised August XX, 2022. This work has been funded in part by the European Fund for Regional Development FEDER and by Junta de Andalucia (project P18-RT-3175). The review of this paper was coordinated by Prof. Kai Zeng.}
	
	\blfootnote{\noindent F.R. Ghadi and F.J. L{\'o}pez Mart{\'i}nez are  with the Communications and Signal Processing Lab, Telecommunication Research Institute (TELMA), Universidad de M{\'a}laga, M{\'a}laga, 29010, (Spain). F.J. L{\'o}pez Mart{\'i}nez is also with the Dept. Signal Theory, Networking and Communications, University of Granada, 18071, Granada (Spain). (e-mail: $\rm farshad@ic.uma.es$, $\rm fjlm@ugr.es$).}
	
	\blfootnote{\noindent W.P. Zhu is  with Department of Electrical and Computer Engineering, Concordia University, Montreal, QC H3G 1M8, Canada (e-mail: $\rm weiping@ece.concordia.ca$).}
	
	\blfootnote{\noindent Jean-Marie Gorce is with the laboratoire CITI
	(a joint laboratory between the Universit\'e de Lyon, INRIA, and INSA de Lyon). 6 Avenue des Arts, F-69621, Villeurbanne, France (e-mail: $\rm jean-marie.gorce@inria.fr$).}
	\blfootnote{Digital Object Identifier 10.1109/XXX.2021.XXXXXXX}
	\vspace{-3mm}
	\section{Introduction}\label{introduction}
The broadcast nature and the inherent randomness of wireless channels have always made information security vulnerable to eavesdropping and jamming attacks. Therefore, the issues of reliability and security are momentous challenges in designing future wireless networks such as sixth-generation (6G) technology \cite{matthaiou2021road,mucchi2021physical,porambage2021roadmap}. One of the alternative approaches to protect information from unauthorized access and guarantee secure communication instead of applying traditional cryptographic algorithms 
is physical layer security (PLS). The principle of PLS was first proposed by Shannon \cite{shannon1949communication} and then studied by Wyner \cite{wyner1975wire} for a basic wiretap channel exploiting Shannon's notion of perfect secrecy. In Wyner's proposed model, a transmitter wants to send a confidential message to a legitimate receiver through a discrete memoryless channel (DMC) while the eavesdropper attempts to decode the message. In this state, Wyner defined the maximum rate of reliable communication from the transmitter to the legitimate receiver as the secrecy capacity (SC), and showed that perfect secrecy could be achieved when the legitimate receiver has a better channel than the eavesdropper. Later, Csisz{\'a}r and K{\"o}rner extended Wyner's result to non-degraded broadcast channels with confidential messages and showed that positive SC is always achievable if the main channel (transmitter-to-legitimate receiver) is less noisy \cite{csiszar1978broadcast}. Leung-Yan-Cheong and Hellman also generalized Wyner's studies to the Gaussian wiretap channel
and defined the SC as the difference between capacities of
the main and eavesdropper channels (transmitter-to-eavesdropper) \cite{leung1978gaussian}. On the other hand, due to the provable effects of the side information (SI) in reducing destructive effects of the interference and guaranteeing reliable communication with higher rates, Mitrpant \textit{et al.} \cite{mitrpant2006achievable} analyzed the Gaussian wiretap channel with SI. Besides, by considering the Gaussian channel with known SI at the transmitter, Costa \cite{costa1983writing} described a dirty paper model in order to examine how much information can be reliably sent, assuming that the recipient cannot distinguish between ink and dirt. Costa's analysis showed that the dirty channel model has the same capacity as the Gaussian model so that SI does not affect the channel capacity. Exploiting the similar approach of writing on dirty paper, Chen and Vink \cite{chen2008wiretap} studied the impact of SI on the Gaussian wiretap channel to find out how much secret information can be reliably and securely sent to the legitimate receiver without leaking information about the secret message to the eavesdropper. They showed that the SI at transmitter provides a larger SC and guarantees a more secure communication. Moreover, Chia and El Gamal \cite{chia2012wiretap} provided a lower bound of SC for the wiretap channel with SI available causally at both encoder and decoder, where  they showed the lower bound in this case is strictly larger
than that for the non-causal case obtained by Liu and Chen \cite{liu2007wiretap}. \textcolor{blue}{However, it is worth noting that the assumption of non-causally available SI at the transmitter is inherent to the dirty paper channel introduced by Costa \cite{costa1983writing}, and is a key building block in information theory. According to Jafar's work \cite{jafar2006capacity}, capacity advantage of non-causal SI overcausal SI is bounded by the number of \textit{genie bits} required to make the transmitter SI available to the receiver as well, which can be none in some cases, e.g., when SI at the transmitter is a deterministic function of the SI at the receiver. \cite{caire1999capacity}. In any case, the consideration of non-causal SI is of interest for bench-marking purposes as a reference scenario, compared to the case with causal SI.}

From the physical layer viewpoint, the received signal at the legitimate receiver and eavesdropper are different due to propagation environment effects such as large-scale and small-scale fading. Furthermore, the main channel and eavesdropper channel are practically correlated due to the physical limitation of antenna spacing or one physical environment, the proximity of the legitimate receiver and eavesdropper, and the presence or absence of scatters around them \cite{shiu2000fading}. For this purpose, several works have analyzed secrecy metrics of PLS over various correlated fading \textit{blank} (i.e., without SI) wiretap channels in recent years: The upper bound of SC and the asymptotic behavior of outage probability for a correlated Rayleigh fading wiretap channel were respectively studied in \cite{jeon2011bounds}, \cite{zhu2013effects}; closed-form expressions for the average secrecy capacity (ASC) and secrecy outage probability (SOP) over correlated Log-normal fading channels were obtained in \cite{pan2015physical}; SOP performance over correlated composite Nakagami-$m$/Gamma fading channels including shadowing and multi-path fading was investigated in \cite{alexandropoulos2017secrecy}; compact expressions for the ASC and SOP in correlated $\alpha-\mu$ fading channels were derived in \cite{mathur2019secrecy}. Moreover, only recently, the effect of fading correlation on the performance of PLS by exploiting copula theory was addressed in \cite{ghadi2020copula} and \cite{besser2020bounds} again for the case of blank wiretap channels. Copula theory is a plausible approach to incorporate arbitrary dependence structures, that has recently become quite popular in the context of performance analysis of wireless communication systems \cite{9464253,gholizadeh2015capacity,huang2015copula,ghadi2020copula1,zheng2019copula,jorswieck2020copula}. In \cite{ghadi2020copula}, the authors derived closed-form expressions for the ASC, SOP, and secrecy coverage region (SCR) over correlated Rayleigh fading channel using Farlie-Gumbel-Morgenstern (FGM) copula, while the authors in \cite{besser2020bounds} represented the upper and lower bound of SOP for the same fading channel. 

Motivated by the significant role of SI in providing a larger SC and rate equivocation region in wiretap channels, in this paper, we extend the dirty paper model of the Gaussian wiretap channel considered in \cite{chen2008wiretap}, to wireless fading communications, in order to understand the behavior of important secrecy performance metrics under the assumption of SI at the transmitter. 
We also generalize \cite{ghadi2020copula} by considering a \textit{correlated} wiretap fading channel and generate arbitrary multivariate coefficients of the main and eavesdropper channels by copula theory. Then, considering non-causally known SI at the transmitter under correlated Rayleigh fading wiretap channel, we exemplify how closed-form expressions for the ASC and SOP can be obtained for the case of using the FGM copula.  
Therefore, the main contributions of our work are summarized
as follows:

\textbullet~We provide an information-theoretical copula-based formulation of the secure communication model over correlated wireless fading channels assuming non-casually SI at the transmitter and review the concept of copula theory and corresponding points.

\textbullet~We represent a general formulation to describe the arbitrary dependence between fading channels coefficients and corresponding signal-to-noise ratios (SNRs). Then, by exploiting the compact probability density function (PDF) obtained by the FGM copula for Rayleigh fading coefficients, we derived the closed-from expression of ASC and SOP.

\textbullet~Finally, we analyze the impact of non-casually known SI at the transmitter on the performance of ASC and SOP by changing the SI values and the dependence parameter within the defined range. 

The rest of this paper is organized as follows. Section \ref{system-model} describes the system model considered in our work. In section \ref{sec-secrecy capacity}, we characterize the SC of our proposed model. In section \ref{sec-copula}, we briefly review the concept of copula and provide the copula-based multivariate distribution, and then present the main results of secure communication with SI under correlated Rayleigh fading wiretap channel. We derive the exact closed-form expression of ASC and SOP in subsections \ref{subsec-asc} and \ref{subsec-sop}, respectively. In section \ref{sec-results}, the validity of analytical results is illustrated numerically. Finally, the conclusions are drawn in section \ref{sec-conclusion}.
	\section{System Model}\label{system-model}
	\begin{figure}[!t]\vspace{0ex}
		\centering
		\includegraphics[width=.9\columnwidth]{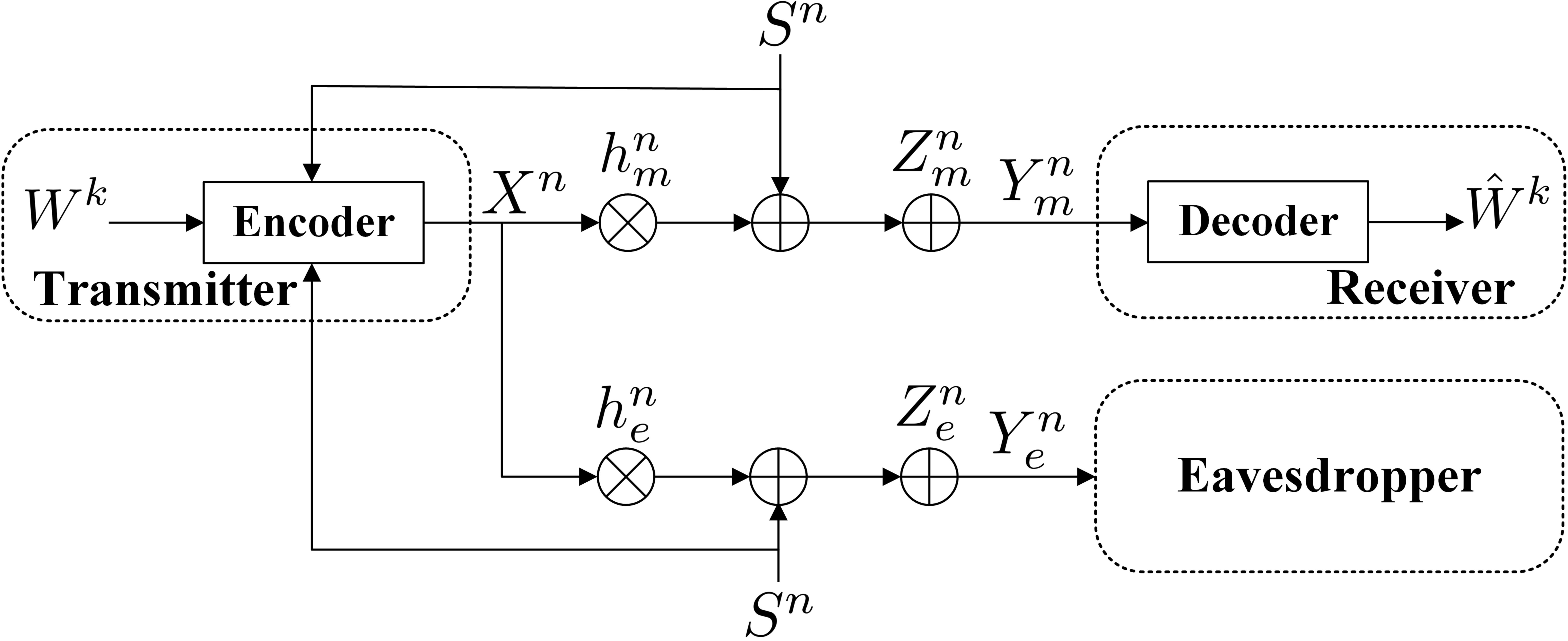} 
		\caption{System model depicting a correlated  fading
			wiretap channel with SI at the transmitter.} 
		\label{model}
	\end{figure}
We consider a wireless wiretap channel with the non-causally known SI $S_i$, $1\leq i\leq n$ at the transmitter\footnote{\textcolor{blue}{We assume that the interfering sequence $S_i$ available as SI at the transmitter is injected from an external dominant
source that exhibits a strong line-of-sight condition with reduced fading fluctuation \cite[Fig. 5]{saad2019vision}. Hence, we use this unfaded counterpart approximation for the fading coefficients corresponding to the interfering signals in our analysis.}} as shown in Fig. \ref{model}. In this scenario, the legitimate transmitter (Alice) sends the confidential message to the legitimate receiver (Bob) through the main channel (transmitter-to-receiver), while the eavesdropper (Eve) attempts to decode the message from its received signal through the eavesdropper channel (transmitter-to-eavesdropper). Let $\mathcal{X}^n$ be the input set, $\mathcal{Y}_m^n$ be the output set of the legitimate receiver, $\mathcal{Y}_e^n$ be the output set of the eavesdropper, and $\mathcal{S}$ be the finite set of SI at the transmitter. Specifically, we assume the transmitter wants to send a message $W^k\in\mathcal{W}^k=\{1,2,...,M\}$ to the legitimate receiver in $n$ uses of the channel. Since the SI is considered to be non-causally known at the transmitter, Alice encodes $W^k$ and $S^n\in\mathcal{S}^n$ into a codeword $X^n\in\mathcal{X}^n$ for transmission over the main channel in a dirty paper fashion. Bob decodes the received signal $Y^n_m\in\mathcal{Y}^n_m$ by making an estimate $\hat{W}^k\left(Y^n_m\right)$ of the message $W^k$, while Eve received the signal $Y^n_e\in\mathcal{Y}^n_e$. Therefore,  the received signals by Bob $Y_m(i)$ and by Eve $Y_e(i)$ can be determined as follows: 
	\begin{align}\label{eq-ym}
		Y_m(i)=h_m(i)X(i)+S(i)+Z_m(i),
	\end{align}
	\begin{align}\label{eq-ye}
		Y_e(i)=h_e(i)X(i)+S(i)+Z_e(i), \quad i=1,...,n,
	\end{align}
where 
$S(i)$ are the non-causally known SI at the transmitter, with variances $Q$ ($S\sim\mathcal{N}\left(0,Q\right)$) which are independently and identically distributed (i.i.d) with probability distribution $p(s)$. The terms $Z_m(i)$ and $Z_e(i)$ correspond to i.i.d. additive white Gaussian noise (AWGN) with zero mean and variances $N_m$ and $N_e$ at the legitimate receiver and eavesdropper, respectively. \textcolor{blue}{Finally, $h_m(i)$ and $h_e(i)$ are the corresponding block fading channel 
 coefficients and hence, the channel power gains are 
  $g_m(i)=|h_m(i)|^2$ and $g_e(i)=|h_e(i)|^2$, where the fading process $h_m(i)$ and $h_e(i)$ are correlated.} We suppose that both the main channel and eavesdropper channel are quasi-static fading channels, meaning that the fading coefficients, albeit random, are constant during the transmission of an entire codeword (i.e., $h_m(i)=h_m$ and $h_e(i)=h_e, \forall i=1,...,n$) and  independent from codeword to
codeword. We also assume that the channel input, the channel fading
coefficients, and the channel noises are all independent. Furthermore, we consider that the codewords sent by transmitter over the channels are subject to the average power constraint as follows:
\begin{align}
\frac{1}{n}\sum_{i=1}^{n}\mathbb{E}\{|X(i)|^2\}\leq P.
\end{align}
Therefore, the instantaneous random  SNR at the legitimate receiver and eavesdropper are given by $\gamma_m=\frac{P|h_m|^2}{N_m}$ and $\gamma_e=\frac{P|h_e|^2}{N_e}$, respectively, while their corresponding  average value are defined as  $\bar{\gamma}_m=\frac{P\mathbb{E}\{|h_m|^2\}}{N_m}$ and $\bar{\gamma}_e=\frac{P\mathbb{E}\{|h_e|^2\}}{N_e}$. 

In the proposed system model, it is assumed that $W^k$ is uniformly distributed on $\{1,2,...,M\}$ as in \citen{chen2008wiretap}. Thus, the transmission rate to the legitimate receiver is defined as $R=H(W^k)/n=\log M/n$, and the equivocation rate of the eavesdropper which indicates the secrecy level of confidential messages against the eavesdropper is denoted as $R_{eq}=H\left(W^k|Y^n_e\right)/n$, where $H\left(W^k|Y^n_e\right)$ is the remaining entropy of $W^k$ conditioned on the  value of $Y^n_e$. Besides, the average error probability is defined as:
\begin{align}
\mathcal{P}_{E}=\frac{1}{M}\sum_{i=1}^{M}\Pr\left(\hat{W}^k\left(Y^n_m\right)\neq|W^k=i\right).
\end{align}

Consequently, the secrecy rate $R_s$ is defined to be achievable, if there exist a code $\left(2^{nR_s},n\right)$ such that for all $\epsilon>0$ and sufficiently large $n$, $\mathcal{P}_{E}\leq\epsilon$ and $R_{eq}\ge R_s-\epsilon$. So, the SC $\mathcal{C}_s$ can be defined as the supremum of achievable transmission rate $R_s$:
\begin{align}
\mathcal{C}_s\overset{\Delta}{=}\underset{\mathcal{P}_E\leq\epsilon}{\text{sup}} R_s.
\end{align}
\section{Secrecy capacity definition}\label{sec-secrecy capacity}
In this section, first, we introduce the SC definition for the discrete memoryless wiretap channel (DMWC) with SI at the transmitter. Then, we exploit this definition to derive the SC of the considered model in a wireless scenario. 
\begin{theorem}\label{thm-inf-cs}
The SC bounds for the DMWC with non-causally known SI at the transmitter was defined in \cite{chen2008wiretap} as follows:
\begin{align}
R_s^{l}\leq\mathcal{C}_s\leq R_s^{u},
\end{align}
where
\begin{align}
R_s^l=\underset{U\rightarrow(X,S)\rightarrow Y_m\rightarrow Y_e}{\max} I(U,Y_m)-\max\{I(U;S),I(U;Y_e)\},
\end{align}
\begin{align}\label{eq10}
R_s^u=\min\{\mathcal{C}_m,R_e\},
\end{align}
\begin{align}
\textcolor{blue}{R_e=\underset{U\rightarrow(X,S)\rightarrow Y_m\rightarrow Y_e}{\max}\left[I(U;Y_m)-I(U;Y_e)\right]},
\end{align}
$U$ is a auxiliary random variable so that $U\rightarrow(X,S)\rightarrow Y_m\rightarrow Y_e$ forms a Markov chain, \textcolor{blue}{$R_e$ is an ancillary secrecy rate, }and \textcolor{blue}{$\mathcal{C}_m$} is the capacity of the main channel with consideration of SI that is defined as:
\begin{align}
\mathcal{C}_m=\underset{U\rightarrow(X,S)\rightarrow Y_m\rightarrow Y_e}{\max} \left[I(U;Y_m)-I(U;S)\right].
\end{align}
\end{theorem}
 \textcolor{blue}{Theorem \ref{thm-inf-cs} shows the bounds for the SC of DMWC in the presence of non-causally known SI at the transmitter. In contrast to the method of proving the SC of the blank DMWC (i.e., the difference between the main and eavesdropper channels),  in which only two extreme points were proved for the achievable region \cite{leung1978gaussian}, this technique provides a more general case of proving achievable region by defining an auxiliary parameter $U$ \cite{chen2008wiretap}. Thus, in order to determine the exact expression of the SC for the corresponding wiretap channel in some specific scenarios of considering SI, we introduce the following corollaries.}
\begin{corollary}\label{col-cm}
According to Theorem \ref{thm-inf-cs}, if there is an auxiliary parameter $U_m$ so that\\
1) $U_m\rightarrow(X,S)\rightarrow Y_m\rightarrow Y_e$ forms a Markov chain,\\
2) $I(U_m;Y_m)-I(U_m;S)=\mathcal{C}_m$,\\
3) $I(U_m;S)\ge I(U_m;Y_e)$,\\
then the SC $\mathcal{C}_s$ is equal to $\mathcal{C}_m$.
\end{corollary}
\begin{corollary}\label{col-cs}
	According to Theorem \ref{thm-inf-cs}, if there is an auxiliary parameter $U_e$ so that\\
	1) $U_e\rightarrow(X,S)\rightarrow Y_m\rightarrow Y_e$ forms a Markov chain,\\
	2) $I(U_e;Y_m)-I(U_e;Y_e)=R_e$,\\
	3) $I(U_e;Y_e)\ge I(U_e;S)$,\\
	then the SC $\mathcal{C}_s$ is equal to $R_e$.
\end{corollary}
\begin{proof}
\textcolor{blue}{The details of the proof are in Appendix \ref{app-coll}.}
\end{proof}
According to \eqref{eq10} and the conditions in Corollaries \ref{col-cm} and \ref{col-cs}, the SC for a DMWC could be $\mathcal{C}_m$ or $R_e$. Now, by extending Theorem \ref{thm-inf-cs} to the Gaussian wiretap channel and assuming the auxiliary random variable $U=X+\alpha S$, where $\alpha$ is a real number and $X$ is independent of $S$, the SC can be determined. \textcolor{blue}{Since the information rate from the transmitter to the receiver must be larger than $0$ for practical reasons, the third condition in both Corollaries \ref{col-cm} and \ref{col-cs} should be considered.} To this end, the following lemma is considered to specify the range of parameter $\alpha$ \textcolor{blue}{and find the optimal value of $\alpha$ for a transmission with maximum possible rate.}
\begin{lemma}\label{lemma1}
The conditions for Corollaries \ref{col-cm} and \ref{col-cs} are explicitly realized based on the system model parameters as:
\begin{align}
&I(U;S)\ge I(U;Y_e)\Longleftrightarrow \alpha\ge\alpha_0\;\;\text{or}\;\;\alpha\le\alpha_{-0},\\
&I(U;S)< I(U;Y_e)\Longleftrightarrow \alpha_{-0}<\alpha<\alpha_0,
\end{align}
where 
\begin{align}
&\alpha_0=\frac{P}{P+N_e}\left(1+\sqrt\frac{P+Q+N_e}{N_e}\right),\\
&\alpha_{-0}=\frac{P}{P+N_e}\left(1-\sqrt\frac{P+Q+N_e}{N_e}\right).
\end{align}
\begin{proof}
The details of the proof are in Appendix \ref{app-lemma1}. 
\end{proof}
\end{lemma}
\textcolor{blue}{From an information theoretic perspective, the result in Theorem \ref{thm-inf-cs} can be also extended to the continuous alphabets and average input constraints by using the standard argument as in \cite{gallager1968information}, and also can be extended to the wireless channels by considering the propagation environment effects such as fading, shadowing, path-loss, etc \cite{el2011network}.} Hence, by considering Lemma \ref{lemma1} and Corollaries \ref{col-cm} and \ref{col-cs}, SC for the block fading wiretap channel with non-causally known SI at the transmitter is defined as the following theorem.
\begin{theorem}\label{thm-sc}
	The SC $\mathcal{C}_s$ for the block fading wiretap channel with non-causally known SI at the transmitter is determined as:
	\begin{equation}\label{Css}
		\mathcal{C}_s=
		\begin{cases}
			\log_2\left(1+\gamma_m\right), \quad\quad\quad\;\;\;\;\;\;\;\;\ \text{if}\;\; \text{Corollary 1} \\
			\log_2\left(\frac{1+\bar{\gamma}_{ms}+\gamma_m}{1+\bar{\gamma}_{es}+\gamma_e}\right), \quad\quad\quad\;\;\; \text{if}\;\;\text{Corollary 2}
		\end{cases},
	\end{equation}
where $\bar{\gamma}_{ms}=\frac{Q}{N_m}$ and $\bar{\gamma}_{es}=\frac{Q}{N_e}$.
\end{theorem}
\begin{proof}
The details of the proof are in Appendix \ref{app-thm-sc}. 
\end{proof}
\section{Copula-based Multivariate Distribution}\label{sec-copula}
Since the fading channel coefficients are assumed correlated in the considered model, we need to generate the corresponding multivariate distribution of the random SNRs. Therefore, in this section, we first briefly review the concept of copula theory \cite{nelsen2007introduction}, and then we determine the joint PDF of the main channel and eavesdropper SNRs. As mentioned in the literature, copula theory is a flexible approach to generate the multivariate distribution of correlated random variables (RVs) by only using their corresponding marginal distributions. Besides, copulas are defined by a particular dependence parameter which indicates the intensity of dependency for unknown RVs.  
\begin{definition}[Two-dimensional copula]
	Let $\vec{V}=(V_1,V_2)$ be a vector of two RVs with marginal cumulative distribution functions (CDFs) $F(v_b)=\Pr(V_b\leq v_b)$ for $b=1,2$, respectively, and relevant bivariate CDF $F(v_1,v_2)=\Pr(V_1\leq v_1,V_2\leq v_2)$. Then, the copula function $C(u_1,u_2)$ of $\vec{V}=(V_1,V_2)$ defined on the unit hypercube $[0,1]^2$ with uniformly distributed RVs $U_b:=F(v_b)$ for $b=1,2$ over $[0,1]$ is given by
	\begin{align}
		C(u_1,u_2)=\Pr(U_1\leq u_1,U_2\leq u_2).
	\end{align}
\end{definition}
\begin{theorem}[Sklar's theorem]
	Let $F(v_1,v_2)$ be a joint CDF of RVs with margins $F(v_b)$ for $b=1,2$. Then, there exists one copula function $C$ such that for all $v_b$ in the extended real line domain $\bar{R}$,%
	\vspace{-1ex}
	\begin{align}\label{eq-sklar}
		F(v_1,v_2)=C\left(F(v_1),F(v_2))\right).
	\end{align}
\end{theorem}
\begin{corollary}\label{col-pdf}
	By applying the chain rule to \eqref{eq-sklar}, the joint PDF $f(v_1,v_2)$ is derived as: 	
	\begin{align}\label{pdf-copula}
		f(v_1,v_2)=f(v_1)f(v_2)c\left(F(v_1),F(v_2)\right),
	\end{align}
	where $c\left(F(v_1),F(v_2)\right)=\frac{\partial^2 C(F(v_1),F(v_2))}{\partial F(v_1)\partial F(v_2)}$ is the copula density function and $f(v_b)$ for $b=1,2$ are the marginal PDFs, respectively.
\end{corollary}
\begin{theorem}[Fr{\'e}chet-Hoeffding bounds] For any copula function $C:[0,1]^2\rightarrow[0,1]$ and any $(u_1,u_2)\in[0,1]^2$, the following bounds hold: $C^-\prec C\prec C^+$; where $C_1\prec C_2$ if $C_1(u_1,u_2)\leq C_2(u_1,u_2) \forall (u_1,u_2)\in[0,1]^2$, and 
	\begin{align}
&C^-(u_1,u_2)=\max\left(u_1+u_2-1,0\right),\\
&C^+(u_1,u_2)=\min\left(u_1,u_2\right).
	\end{align}
\end{theorem}
The upper and lower bounds model extreme dependence structures. If $C=C^-$ the pair of RVs are said to be countermonotonic, whereas $C=C^+$ means that both RVs are comonotonic. The copula function for the independent case $C^\perp (u_1,u_2)=u_1u_2$ defines the limit between positive and negative dependence. Let us assume two Copulas that verify: $C^-\prec C_1\prec C^\perp \prec C_2 \prec C^+$. Then, $C_1$ models a negative dependence and $C_2$ a positive dependence.

\textcolor{blue}{Note that the above definitions are valid for any arbitrary choice of fading distributions as well as copula functions. We will now indicate how the secrecy performance metrics can be characterized in the closed-form expression under correlated Rayleigh fading channels, exploiting the FGM copula. The main advantage of FGM copula is that it 
can describe both positive and negative dependencies between RVs, 
	while offering good mathematical tractability.}

\begin{definition}\label{def}[FGM copula] The FGM copula with dependence parameter $\theta\in[-1,1]$ is defined as:
	\begin{align}\label{eq-fgm}
		C_F(u_1,u_2)=u_1u_2\left(1+\theta(1-u_1)(1-u_2)\right),
	\end{align}
	where $\theta\in[-1,0)$ and $\theta\in(0,1]$ denote the negative and positive dependence structures respectively, while $\theta=0$ indicates the independence structure.
	\end{definition}
		\textcolor{blue}{Since we assume the channel fading coefficients follow Rayleigh distribution, the corresponding random SNRs $\gamma_m>0$ and $\gamma_e>0$ are exponentially distributed with following marginal distributions:  $f(\gamma_j)=\frac{1}{\bar{\gamma}_j}\mathrm{e}^{-\frac{\gamma_j}{\bar{\gamma}_j}},F(\gamma_j)=1-\mathrm{e}^{-\frac{\gamma_j}{\bar{\gamma}_j}}$ for $j\in\{m,e\}$, respectively. }
\begin{lemma}\label{lemma-pdf}
	The joint PDF of $\gamma_m$ and $\gamma_e$ based on FGM copula is determined as:
	\begin{align}
	f(\gamma_m,\gamma_e)=\frac{e^{-\frac{\gamma_m}{\bar{\gamma}_m}-\frac{\gamma_e}{\bar{\gamma}_e}}}{\bar{\gamma}_m\bar{\gamma}_e}\Big[1+\theta\big(1-2e^{-\frac{\gamma_m}{\bar{\gamma}_m}}\big)\big(1-2e^{-\frac{\gamma_e}{\bar{\gamma}_e}}\big)\Big].
	\end{align}
\end{lemma}
\begin{proof}
By applying the partial derivatives in \eqref{eq-fgm}, inserting it into \eqref{pdf-copula}, and then considering the SNRs marginal distribution, the proof is completed. 
\begin{Remark}
	\textcolor{blue}{It is possible incorporate the effect of fading severity into the analysis, i.e., to consider more general fading conditions than Rayleigh fading. This usually comes at the price of an increased mathematical complexity. For instance, assuming Nakagami-$m$ fading would require the use of a generalization of the correlated Gamma distribution (i.e. more complex than the one in \cite[Theorem 4]{ghadi2022capacity}, to include different fading severity parameters for Bob and Eve's channels. 
}
	\end{Remark}
	\end{proof}
\section{Secrecy Performance Metrics: ASC and SOP}
In this section, we derive closed-form expressions for ASC and SOP by exploiting the joint PDF obtained in section \ref{sec-copula}.
\subsection{ASC Analysis}\label{subsec-asc}
The ASC for considered system model based on Corollaries \ref{col-cm} and \ref{col-cs} can be defined as:
\begin{align}\label{eq-sc1}
\bar{\mathcal{C}}_s^{1}=\int_{0}^{\infty}\int_{0}^{\infty}\log_2\left(1+\gamma_m\right) f\left(\gamma_m,\gamma_{e}\right)d\gamma_{e}d\gamma_m,
\end{align}
\begin{align}\label{eq-sc2}
\bar{\mathcal{C}}_s^{2}=\int_{0}^{\infty}\int_{0}^{\bar{\gamma}}\log_2\left(\frac{1+\bar{\gamma}_{ms}+\gamma_m}{1+\bar{\gamma}_{es}+\gamma_{e}}\right) f\left(\gamma_m,\gamma_{e}\right)d\gamma_{e}d\gamma_m,
\end{align}
where $\bar{\gamma}=\gamma_m+\bar{\gamma}_{ms}-\bar{\gamma}_{es}$. 
\begin{theorem}\label{thm-asc}
	The ASC for concerned correlated Rayleigh fading wiretap channel with non-causally known SI at the transmitter under consideration of the Corollary \ref{col-cm} and Corollary \ref{col-cs} is determined as:
	\begin{equation}\label{}
		\bar{\mathcal{C}}_s=
		\begin{cases}
			\bar{\mathcal{C}}_s^1, \quad\quad\;\;\ \text{if}\;\; \text{Corollary 1} \\
			\bar{\mathcal{C}}_s^2, \quad\quad\;\;\; \text{if}\;\;\text{Corollary 2}
		\end{cases},
	\end{equation}
where $\bar{\mathcal{C}}_s^1$ and $\bar{\mathcal{C}}_s^2$ are given by 	 \eqref{eq-asc1} and \eqref{eq-asc2}, respectively; and $\mathrm{E}_1(t)=\int_{t}^{\infty}\frac{\mathrm{e}^{-z}}{z}dz$ is the Exponential Integral.
%
	\begin{figure*}[!t]
		\normalsize
		\begin{align}\label{eq-asc1}
	\bar{\mathcal{C}}_s^1=\frac{\mathrm{e}^{\frac{1}{\bar{\gamma}_m}}}{\ln 2}\mathrm{E}_1\left(\frac{1}{\bar{\gamma}_m}\right),
		\end{align}
		\begin{align}\nonumber\label{eq-asc2}
&\bar{\mathcal{C}}_s^{2}=\frac{1}{\ln 2}\Bigg(\ln\left(\frac{1+\bar{\gamma}_{ms}}{1+\bar{\gamma}_{es}}\right)+\mathrm{e}^\frac{1+\bar{\gamma}_{ms}}{\bar{\gamma}_m}\mathrm{E}_1\left(\frac{1+\bar{\gamma}_{ms}}{\bar{\gamma}_m}\right)-\mathrm{e}^{\frac{1+\bar{\gamma}_{es}}{\bar{\gamma}_{e}}}\left(\mathrm{E}_1\left(\frac{1+\bar{\gamma}_{ms}}{\bar{\gamma}_e}\right)-\mathrm{E}_1\left(\frac{1+\bar{\gamma}_{es}}{\bar{\gamma}_e}\right)\right)\\\nonumber
&-\mathrm{e}^{\frac{1+\bar{\gamma}_{es}}{\bar{\gamma}_e}+\frac{1+\bar{\gamma}_{ms}}{\bar{\gamma}_m}}\Bigg(\mathrm{E}_1\left(\frac{(\bar{\gamma}_m+\bar{\gamma}_e)(1+\bar{\gamma}_{ms})}{\bar{\gamma}_m\bar{\gamma}_e}\right)+\theta\Bigg[\mathrm{E}_1\left(\frac{(\bar{\gamma}_m+\bar{\gamma}_e)(1+\bar{\gamma}_{ms})}{\bar{\gamma}_m\bar{\gamma}_e}\right)-\mathrm{e}^{\frac{(1+\bar{\gamma}_{ms})}{\bar{\gamma}_m}}\mathrm{E}_1\left(\frac{(\bar{\gamma}_m+2\bar{\gamma}_e)(1+\bar{\gamma}_{ms})}{\bar{\gamma}_m\bar{\gamma}_e}\right)\\
&-\mathrm{e}^{\frac{(1+\bar{\gamma}_{es})}{\bar{\gamma}_e}}\mathrm{E}_1\left(\frac{(2\bar{\gamma}_m+\bar{\gamma}_e)(1+\bar{\gamma}_{ms})}{\bar{\gamma}_m\bar{\gamma}_e}\right)+\mathrm{e}^{\frac{(1+\bar{\gamma}_{es})}{\bar{\gamma}_e}+\frac{(1+\bar{\gamma}_{ms})}{\bar{\gamma}_m}}\mathrm{E}_1\left(\frac{2(\bar{\gamma}_m+\bar{\gamma}_e)(1+\bar{\gamma}_{ms})}{\bar{\gamma}_m\bar{\gamma}_e}\right)\Bigg]\Bigg)\Bigg).
		\end{align}
				\hrulefill\vspace{0ex}
		\vspace*{0pt}
		\end{figure*}
	\end{theorem}
\begin{proof}
The details of the proof are in Appendix \ref{app-thm-asc}.
\end{proof}
\subsection{SOP Analysis}\label{subsec-sop}
The SOP is defined as the probability that the random SC $\mathcal{C}_s$ is less than a target secrecy rate $R_s>0$, or:
	\begin{align}\label{eq-sop-def}
		P_{sop}=\mathrm{Pr}(\mathcal{C}_s\le R_s)=1-\mathrm{Pr}(\mathcal{C}_s>R_s).
	\end{align}
So, the SOP under consideration of Corollary \ref{col-cm} is defined as:
	\begin{align}
	P_{sop}^1&=1-\mathrm{Pr}\left(\log_2\left(1+\gamma_m\right)>R_s\right)\\ 
	&=1-\mathrm{Pr}\left(\gamma_m>2^{R_s}-1\right)\\
	&=1-\int_{0}^{\infty}\int_{2^{R_s}-1}^{\infty}f(\gamma_m,\gamma_e)d\gamma_md\gamma_e,\label{eq-sop-1}
\end{align}
and similarly, the SOP under consideration of Corollary \ref{col-cs} is given by
\begin{align}
	P_{sop}^{2}&=1-\mathrm{Pr}\left(\log\left(\frac{1+\bar{\gamma}_{ms}+\gamma_m}{1+\bar{\gamma}_{es}+\gamma_e}\right)>R_s\right)\\ 
	&=1-\mathrm{Pr}\left(\gamma_m>2^{R_s}\left(1+\bar{\gamma}_{es}+\gamma_e\right)-\left(1+\bar{\gamma}_{ms}\right)\right)\\
	&=1-\int_{0}^{\infty}\int_{\gamma_{th}}^{\infty}f(\gamma_m,\gamma_e)d\gamma_md\gamma_e,\label{eq-sop2}
\end{align}
where $\gamma_{th}=2^{R_s}\left(1+\bar{\gamma}_{es}+\gamma_e\right)-\left(1+\bar{\gamma}_{ms}\right)$.
\begin{theorem}\label{thm-sop}
	The SOP for concerned correlated Rayleigh fading wiretap channel with non-causally known SI at the transmitter under consideration of the Corollary \ref{col-cm} and Corollary \ref{col-cs} is given by
	\begin{equation}\label{}
	P_{sop}=
	\begin{cases}
		P_{sop}^1, \quad\quad\;\;\ \text{if}\;\; \text{Corollary 1} \\
		P_{sop}^2, \quad\quad\;\;\; \text{if}\;\;\text{Corollary 2}
	\end{cases},
\end{equation}
where $P_{sop}^1$ and $P_{sop}^2$ are expressed as follows:
\begin{align}
	P_{sop}^1=&\,1-\mathrm{e}^{-\frac{\left(2^{R_s}-1\right)}{\bar{\gamma}_m}},\label{eq-sop1}
\end{align}

\begin{align}\nonumber
	P_{sop}^2=&\,1-\Bigg[\frac{\bar{\gamma}_m\mathrm{e}^{\frac{-\bar{\gamma}_{th}}{\bar{\gamma}_m}}}{\bar{\gamma}_m+2^{R_s}\bar{\gamma}_e}+\theta\Bigg(\frac{\bar{\gamma}_m\mathrm{e}^{\frac{-\bar{\gamma}_{th}}{\bar{\gamma}_m}}}{\bar{\gamma}_m+2^{R_s}\bar{\gamma}_e}-\frac{\bar{\gamma}_m\mathrm{e}^{\frac{-2\bar{\gamma}_{th}}{\bar{\gamma}_m}}}{\bar{\gamma}_m+2^{R_s+1}\bar{\gamma}_e}\\
	&-\frac{2\bar{\gamma}_m\mathrm{e}^{\frac{-\bar{\gamma}_{th}}{\bar{\gamma}_m}}}{2\bar{\gamma}_m+2^{R_s}\bar{\gamma}_e}+\frac{\bar{\gamma}_m\mathrm{e}^{\frac{-2\bar{\gamma}_{th}}{\bar{\gamma}_m}}}{\bar{\gamma}_m+2^{R_s}\bar{\gamma}_e}\Bigg)\Bigg],
\end{align}
and  $\bar{\gamma}_{th}=2^{R_s}\left(1+\bar{\gamma}_{es}\right)-\left(1+\bar{\gamma}_{ms}\right)$.
\end{theorem}
\begin{proof}
The details of proof are in Appendix \ref{app-thm-sop}.
\end{proof}
\section{Numerical Results}\label{sec-results}
In this section, the analytical expressions previously derived and Monte-Carlo (MC) simulation for the ASC and SOP are presented, with the special focus on comparing the performances in the presence/absence of SI and fading correlation. It should be noted that for the case of analyzing ASC under correlated fading, additional simulations are included using the Frank copula \cite{nelsen2007introduction} to extend the analytical results obtained with the FGM copula to a wider range of values for both positive and negative dependencies by exploiting Frank's copula parameter $\zeta\in\mathbb{R}\backslash\{0\}$.
\begin{figure}[!t]\vspace{0ex}
	\centering
	\includegraphics[width=.9\columnwidth]{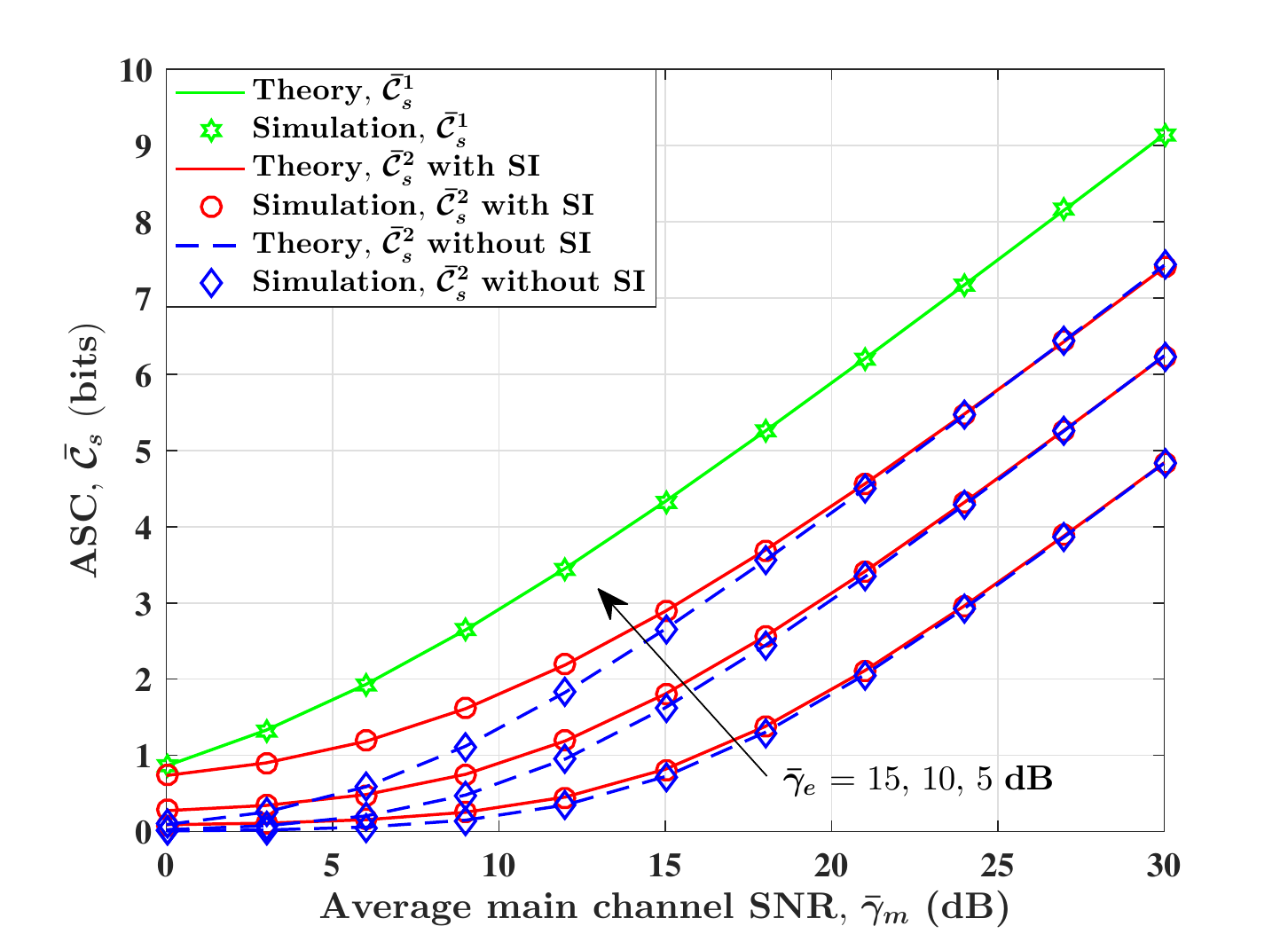} 
	\caption{The ASC versus $\bar{\gamma}_m$ for selected values of $\bar{\gamma}_e$, a target secrecy rate $R_s=1.5$ bits, the FGM dependence parameter $\theta=1$, and given values of SI $\bar{\gamma}_{ms}=5$ dB, $\bar{\gamma}_{es}=-10$ dB.} 
	\label{asc_gm_n}
\end{figure}
Fig. \ref{asc_gm_n} shows the impact of SI on the performance of ASC $\bar{\mathcal{C}}_s$ based on the variations of the average main channel SNR $\bar{\gamma}_m$ for selected values of the average eavesdropper channel SNR $\bar{\gamma}_e$ and given values of SI under positive dependence structure. It can be seen that SI does not affect the performance of the ASC under consideration of Corollary \ref{col-cm} since $\mathcal{C}_s^1$ is independent of SI and only depends on $\gamma_{m}$ based on SC definition in \eqref{Css}. On the other hand, it is observed that the SI has positive effects on the behavior of the ASC under Corollary \ref{col-cs}, so that for the low SNR regime the efficiency of SI is more tangible as compared with the high SNR regime. To further evaluate the impact of SI, the behavior of ASC based on the SI ratio (i.e., $\bar{\gamma}_{ms}/\bar{\gamma}_{es}$) for three different scenarios $\bar{\gamma}_m>\bar{\gamma}_e$, $\bar{\gamma}_m=\bar{\gamma}_e$, and $\bar{\gamma}_m<\bar{\gamma}_e$ under positive dependence structure is illustrated in Fig. \ref{asc_si_ratio_n}. It is clear that for all three scenarios, $\bar{\mathcal{C}}_s^1$ is constant during the changes of the SI ratio and only depends on the values of $\bar{\gamma}_m$. On the other side, $\bar{\mathcal{C}}_s^2$ continuously increases by increasing the SI ratio for all scenarios. The interesting point is that even under the condition that the main channel is worse than the eavesdropper channel (i.e., $\bar{\gamma}_m<\bar{\gamma}_e$), $\bar{\mathcal{C}}_s^2$ still grows and the ASC is achievable. In Fig. \ref{sop_gm_n}, the effect of SI on the performance of SOP based on the variations of $\bar{\gamma}_m$ for selected values of $\bar{\gamma}_e$ and given values of SI under positive dependence structure is provided. It can be seen that $P_{sop}^1$ provides the lowest values of SOP and it is independent of the SI and $\bar{\gamma}_e$. In contrast, considering SI at the transmitter has constructive effects on the performance of $P_{sop}^2$ so that for all selected values of $\bar{\gamma}_e$, a lower value of SOP is achievable as compared to the case where there is no SI. 
\begin{figure}[!t]\vspace{0ex}
	\centering
	\includegraphics[width=.9\columnwidth]{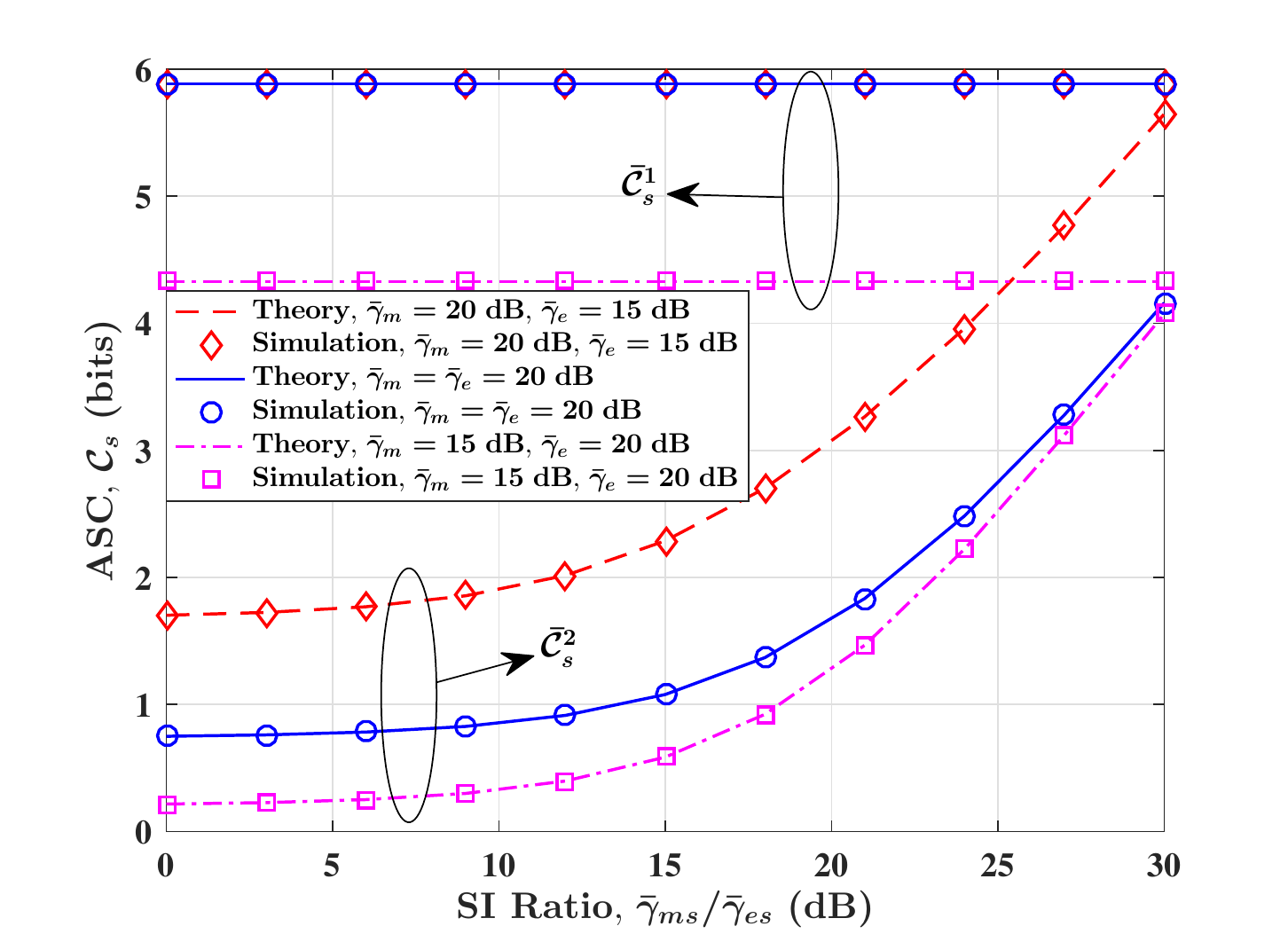} 
	\caption{The ASC versus $\bar{\gamma}_{ms}/\bar{\gamma}_{es}$for a target secrecy rate $R_s=1.5$ bits, the FGM dependence parameter $\theta=1$, and different scenarios when, (a) the main channel's condition is better than the eavesdropper's ($\bar{\gamma}_m>\bar{\gamma}_e$), (b) the main channel's condition is the same as eavesdropper's ($\bar{\gamma}_m=\bar{\gamma}_e$), and (c) the main channel's condition is worse than the eavesdropper's ($\bar{\gamma}_m<\bar{\gamma}_e$).} 
	\label{asc_si_ratio_n}
\end{figure}
\begin{figure}[!t]\vspace{0ex}
	\centering
	\includegraphics[width=.9\columnwidth]{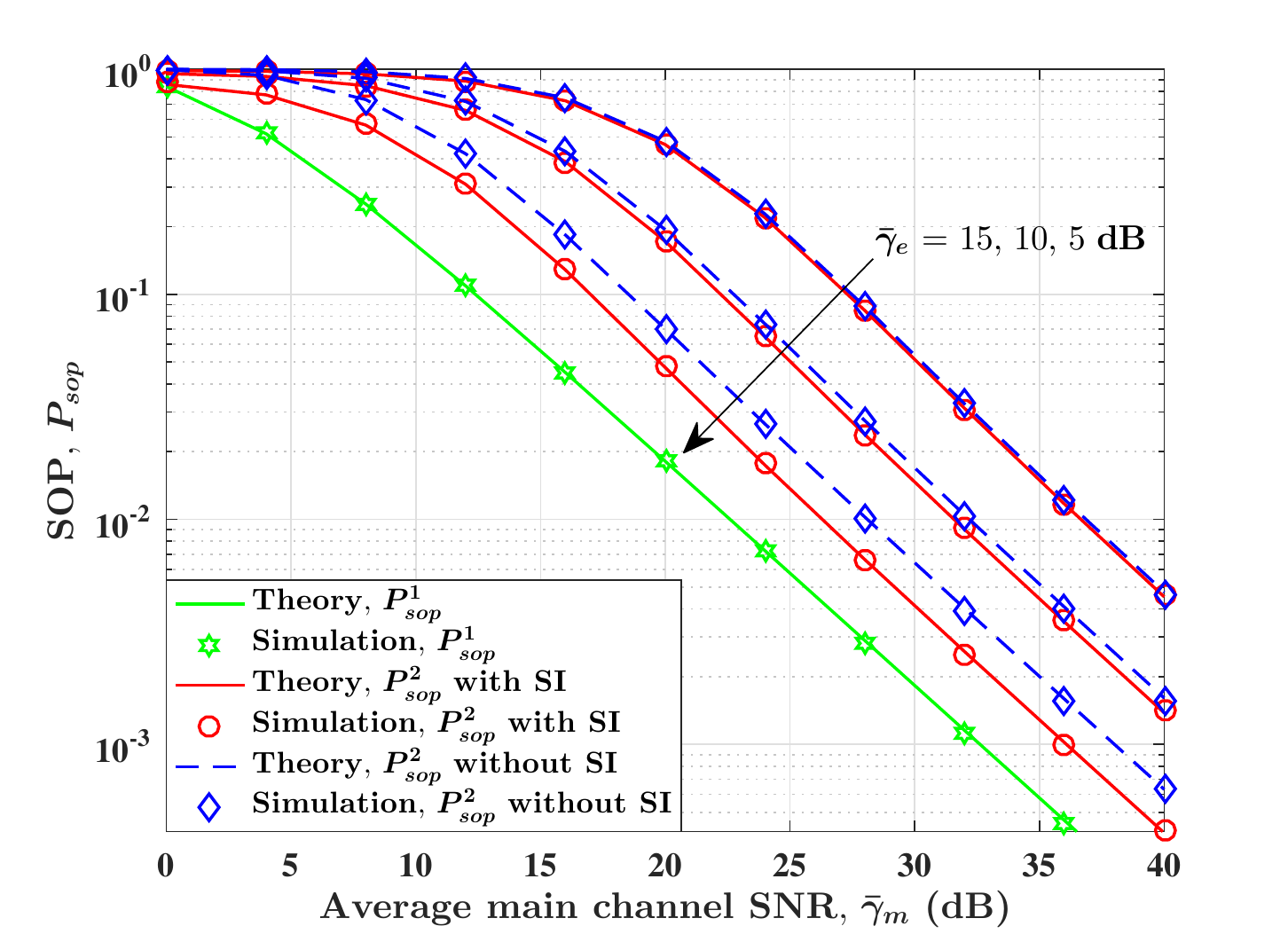} 
	\caption{The SOP versus $\bar{\gamma}_m$ for selected values of $\bar{\gamma}_e$, a target secrecy rate $R_s=1.5$ bits, the FGM dependence parameter $\theta=1$, and given values of SI $\bar{\gamma}_{ms}=5$ dB, $\bar{\gamma}_{es}=-5$ dB.} 
	\label{sop_gm_n}
\end{figure}
From a correlation viewpoint, we evaluate the impact of fading correlation on the performance of ASC and SOP in the presence of SI in Figs. \ref{asc_gm_corr_n} and \ref{sop_gm_corr_n}. In these cases, we have modeled the correlated fading channels based on the copula approach where the dependence parameter shows the measure of dependency. In this regard, it should be noted there is a relation between the linear correlation coefficient $\rho\overset{\Delta}{=}\mathrm{cov}[\gamma_{m}\gamma_{e}]/\sqrt{\mathrm{var}[\gamma_{m}]\mathrm{var}[\gamma_{e}]}$ and the copula dependence parameter. Therefore, to have a clearer insight with conventional correlation, we also approximate the corresponding values of $\rho$ for the considered copula dependence parameters. Fig. \ref{asc_gm_corr_n} shows the behavior of ASC in terms of $\bar{\gamma}_m$ under correlated Rayleigh fading channels for selected values of the dependence parameters and SI. Given the independence of $\bar{\mathcal{C}}_s^1$ from correlation and SI, it is clear that $\bar{\mathcal{C}}_s^1$ gains higher values in terms of ASC during the changes of $\bar{\gamma}_m$. In contrast, it can be observed that the correlation effect is gradually eliminated in the high SNR regime in terms of $P_{sop}^2$.  We also see that positive dependence worsens the ASC performance compared to the independent fading case since as the channels become more correlated, they tend to behave more similarly, and thus the chances for the transmitter to transmit at a high secrecy rate decreases. It should be noted that the Fr{\'e}chet-Hoeffding bounds are not symmetric w.r.t. $\rho=0$, but they tend to $\pm1$. Therefore, we observe that the Frank copula gets close to the Fr{\'e}chet-Hoeffding bounds for $\zeta=\pm35$ in Fig. \ref{asc_gm_corr_n}, which exhibits a non-symmetric behavior. We also see that correlation for the case with FGM copula is now symmetric since such copula only can model weak dependences, and hence is distant to the Fr{\'e}chet-Hoeffding bounds, thus not reaching the maximum permissive range of $\pm1$ for $\rho$.
\begin{figure}[!t]\vspace{0ex}
	\centering
	\includegraphics[width=.9\columnwidth]{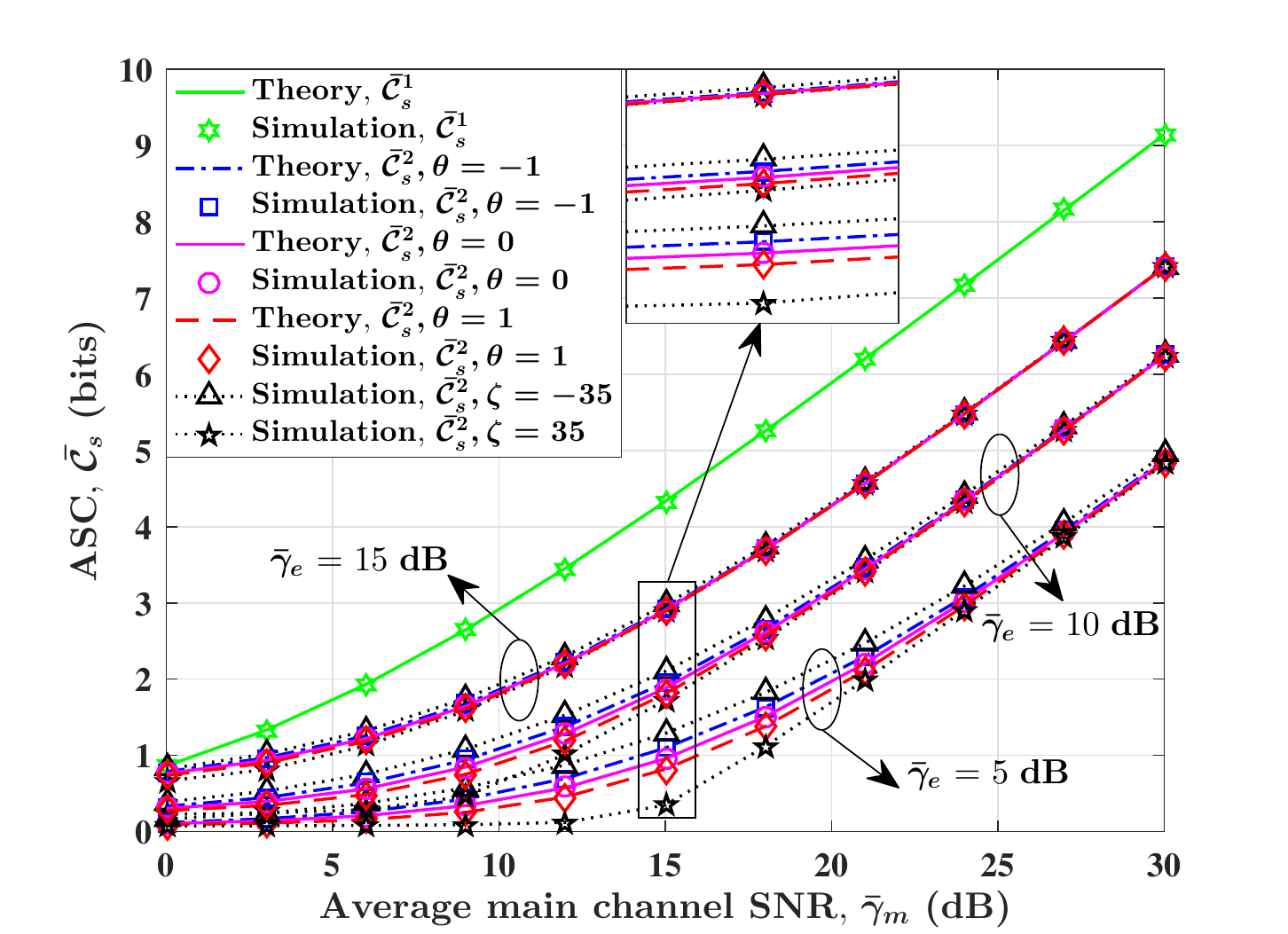} 
	\caption{The ASC versus $\bar{\gamma}_m$ for selected values of $\bar{\gamma}_e$, a target secrecy rate $R_s=1.5$ bits, the FGM dependence parameter $\theta=\pm1$ $(\rho\approx\pm0.25)$, $\zeta=\pm35$ $(\rho\approx[-0.63,0.92])$ and given values of SI $\bar{\gamma}_{ms}=5$ dB, $\bar{\gamma}_{es}=-10$ dB.} 
	\label{asc_gm_corr_n}
\end{figure}
\begin{figure}[!t]\vspace{0ex}
	\centering
	\includegraphics[width=.9\columnwidth]{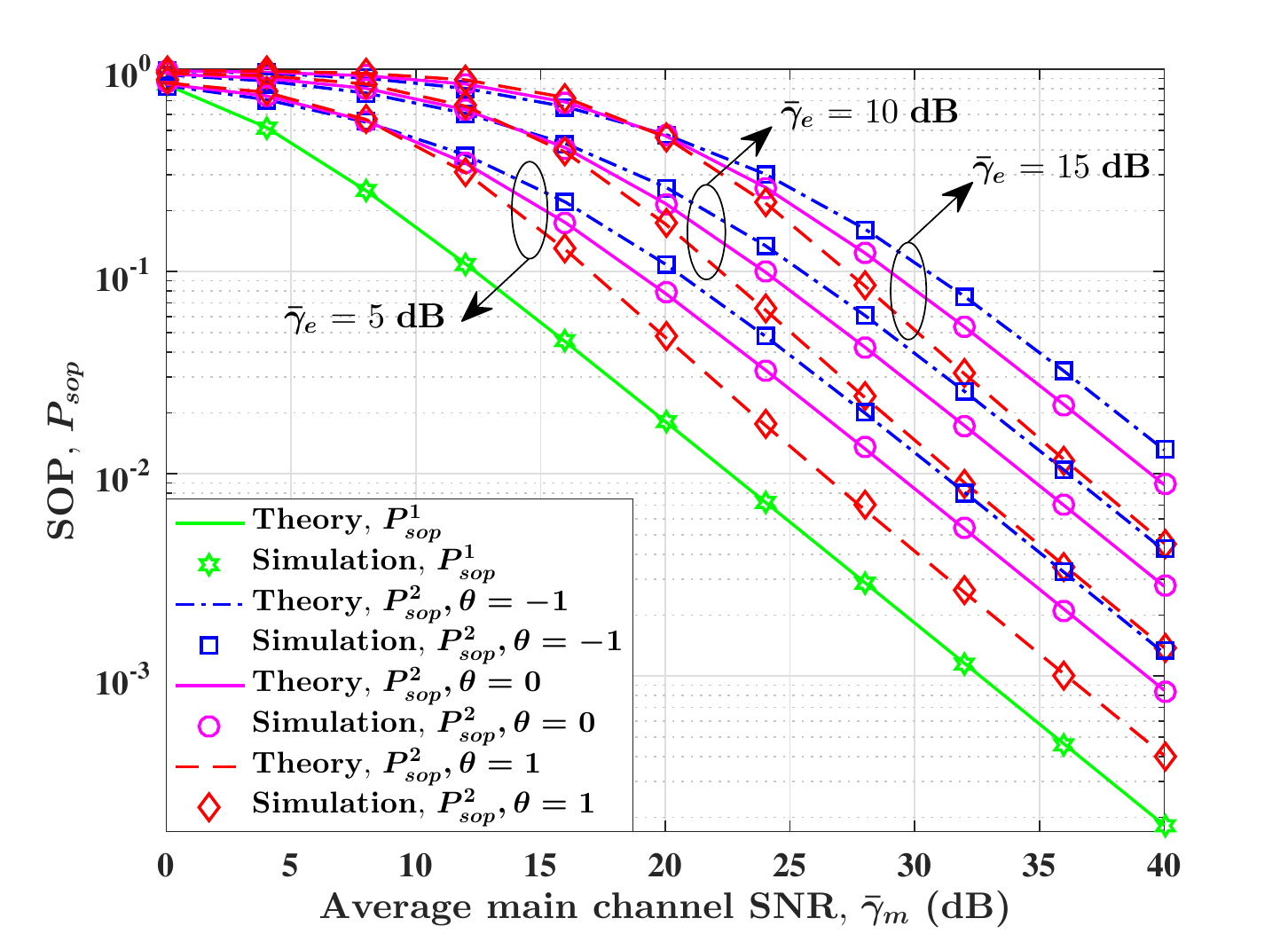} 
	\caption{The SOP versus $\bar{\gamma}_m$ for selected values of $\bar{\gamma}_e$, a target secrecy rate $R_s=1.5$ bits, the FGM dependence parameter $\theta=\pm1$ $(\rho\approx\pm0.25)$, $\zeta=\pm35$ $(\rho\approx[-0.63,0.92])$ and given values of SI $\bar{\gamma}_{ms}=5$ dB, $\bar{\gamma}_{es}=-5$ dB.} 
	\label{sop_gm_corr_n}
\end{figure}

In Fig. \ref{sop_gm_corr_n} the performance of SOP in terms of $\bar{\gamma}_m$ under correlated Rayleigh fading channels for selected values of dependence parameters and SI is represented. As expected, $P_{sop}^1$ provides the lowest values of SOP during $\bar{\gamma}_m$ variations regardless the correlation effects. In contrast, it can be seen that when the main channel is better than the eavesdropper channel ($\bar{\gamma}_m>\bar{\gamma}_e$), $P_{sop}^2$ will be achieved less than $0.5$, while under the condition that $\bar{\gamma}_m\leq\bar{\gamma}_e$, $P_{sop}^2$ becomes greater than $0.5$. In this case, 
we can see that the positive dependence improves the SOP performance as compared with the independent fading case when $P_{sop}^2<0.5$, since if $\bar{\gamma}_m>\bar{\gamma}_e$, then the larger the correlation level, the higher the probability of having $\gamma_m>\gamma_e$. 

\section{Conclusion}\label{sec-conclusion}
In this paper, we analyzed the impact of SI on PLS performance under correlated fading channels. To this end, we derived the closed-from expressions for ASC and SOP by exploiting copula theory in the presence of SI and fading correlation. We proved that considering the SI at the transmitter, compared with the \textit{blank} fading wiretap channel considered in \cite{ghadi2020copula}, improves the efficiency of the system model in terms of ASC and SOP. Hence, an increase in SI provides a higher value of ASC as well as a lower SOP. We also showed that for a fixed SI, an increment in the correlation decreases the ASC in the low SNR regime, while improves the performance of SOP in the specific cases. \textcolor{blue}{The extension of these results to consider different interfering sources affecting Bob and Eve, their impact on the system's SI, the effect of causal SI knowledge, and the consideration of more general fading conditions are potential lines for future research activities.}
\appendices
\section{Proof of Corollaries \ref{col-cm} and \ref{col-cs}}\label{app-coll}
\textcolor{blue}{Regarding the definition of $R_s$ and $R_e$, we have $R_s\leq R_e$. Besides, since $I(U_e;Y_e)\leq I(U_e;S)$ and $U_e$ maximizes $I(U_e;Y_m)-I(U_e;Y_e)$, it follows that:
\begin{align}
	R_e&=I(U_e;Y_m)-I(U_e;Y_e)\\
	&=I(U_e;Y_m)-\max\left\{I(U_e;S)-I(U_e;Y_e)\right\}\leq R_s.
\end{align}
Thus, we have $R_s=R_e$. On the other hand, it is known that $\mathcal{C}_s\geq R_s=R_e$. Hence, it is only needed to prove $C_s\leq R_e$. For this purpose, by exploiting the data processing theorem and Fano's inequality, we have:
\begin{align}
H\left(W^k|Y_m^n,Y_e^n\right)\leq & H\left(W^k|Y_m^k\right)\leq H\left(W^k|\hat{W}^k\right)\\
&\leq \mathrm{h}\left(\mathcal{P}_E\right)+nR\mathcal{P}_E,
\end{align}
where $\mathrm{h}(.)$ is the binary entropy function. Now, using the definitions $R=H\left(W^k\right)/n$ and $d=H\left(W^k|Y_e^n\right)/H\left(W^k\right)$, $nRd$ can be determined as:
\begin{align}
&nRd=H\left(W^k|Y_e^n\right)\\
&\hspace{-0.1cm}\leq H\left(W^k|Y_e^n\right)-H\left(W^k|Y_m^n,Y_e^n\right)+\mathrm{h}\left(\mathcal{P}_E\right)+nR\mathcal{P}_E\\
&\overset{(a)}{\leq}\sum_{i=1}^{n}I\left(W^k;Y_{m_i}|Y_{m_1}^{i-1},Y_e^n\right)+\mathrm{h}\left(\mathcal{P}_E\right)+nR\mathcal{P}_E\\
&\hspace{-0.2cm}\overset{(b)}{=}\sum_{i=1}^{n}I\left(W^k;Y_{m_i}|Y_{m_1}^{i-1},Y_e^n,S_1^{i-1},S_{i+1}^n\right)+\mathrm{h}\left(\mathcal{P}_E\right)+nR\mathcal{P}_E\\
&\overset{(c)}{\leq} \sum_{i=1}^{n}I\left(U_i;Y_{m_i}|Y_{e_i}\right)+\mathrm{h}\left(\mathcal{P}_E\right)+nR\mathcal{P}_E\\
&\overset{(d)}{=}\sum_{i=1}^{n}\left[I\left(U_i;Y_{m_i}\right)-I\left(U_i;Y_{e_i}\right)\right]+\mathrm{h}\left(\mathcal{P}_E\right)+nR\mathcal{P}_E,
\end{align}
where $(a)$  is obtained from the fact of chain rule for information, $(b)$ is followed from the fact that $S^n$ is independent of $W^k$, $(c)$ is achieved from the assumption that $U_i=\left(W^k,Y_{m_1}^{i-1},Y_{e_1}^{i-1},Y_{e_{i+1}}^n,S_1^{i-1},S_{i+1}^n\right)$, and $(d)$ is followed from the fact that $\left(U_i,S_i\right)\rightarrow Y_{m_i}\rightarrow Y_{e_i}$ forms a Markov chain. Now, by choosing $i^*$ to be the index $i$ such that $I\left(U_{i^*};Y_{m_{i^*}}\right)-I\left(U_{i^*};Y_{e_{i^*}}\right)=\underset{l\in\left[1,n\right]}{\max}\left\{I\left(U_l;Y_{m_l}\right)-I\left(U_l;Y_{e_l}\right)\right\}$, we have:
\begin{align}
	Rd&\leq I\left(U_{i^*};Y_{m_{i^*}}\right)-I\left(U_{i^*};Y_{e_{i^*}}\right)+\frac{\mathrm{h}\left(\mathcal{P}_E\right)}{N}+R\mathcal{P}_E\\
	&\leq R_e+\frac{\mathrm{h}\left(\mathcal{P}_E\right)}{N}+R\mathcal{P}_E.
	\end{align}
So, noting that $\mathcal{C}_s$ is the maximum value of $R$ when $d$ approaches to $1$, we have $\mathcal{C}_s\leq R_e+\frac{\mathrm{h}\left(\mathcal{P}_E\right)}{N}+R\mathcal{P}_E$. Consequently, $\mathcal{C}_s\leq R_e$ and the proof is completed.}
\section{Proof of Lemma \ref{lemma1}}\label{app-lemma1}
By considering the Gaussian case, we compute the values of  mutual information $I(U;S)$ and $I(U;Y_e)$ as follows:
\begin{align}
	&I(U;S)=I(X+\alpha S;S)\\
	&=H(X+\alpha S)+H(S)-H(X+\alpha S,S)\\\nonumber
	&=\log\left(\left(2\pi\mathrm{e}\right)^2\left(P+\alpha^2Q\right)Q\right)\\
	&\quad-\log\left(\left(2\pi\mathrm{e}\right)^2\left(\left(P+\alpha^2Q\right)Q-\alpha^2Q^2\right)\right)\\
	&=\log\left(\frac{P+\alpha^2Q}{P}\right).
\end{align}
and 
\begin{align}
	&I(U;Y_e)=I(X+\alpha S;X+S+Z_e)\\
	&=H(X+S+Z_e)-H(X+S+Z_e|X+\alpha S)\\\nonumber
	&=H(X+S+Z_e)+H(X+\alpha S)\\
	&\quad-H(X+S+Z_e,X+\alpha S)\\\nonumber
	&=\log\left(\left(2\pi\mathrm{e}\right)^2\left(P+Q+N_e\right)\left(P+\alpha^2Q\right)\right)\\
	&\quad-\log\left(\left(2\pi\mathrm{e}\right)^2\det\left(\mathrm{cov}\left(X+S+Z_e,X+\alpha S\right)\right)\right)\\
	&=\log\left(\frac{\left(P+Q+N_e\right)\left(P+\alpha^2Q\right)}{PQ\left(1-\alpha\right)^2+N_e\left(P+\alpha^2Q\right)}\right).
\end{align}
Now we consider the inequality $I(U;S)\ge I(U;Y_e)$, and we have:
\begin{align}
&I(U;S)\ge I(U;Y_e)\\
&\log\left(\frac{P+\alpha^2Q}{P}\right)\ge\log\left(\frac{\left(P+Q+N_e\right)\left(P+\alpha^2Q\right)}{PQ\left(1-\alpha\right)^2+N_e\left(P+\alpha^2Q\right)}\right)\\
&\frac{1}{P}>\frac{\left(P+Q+N_e\right)}{PQ\left(1-\alpha\right)^2+N_e\left(P+\alpha^2Q\right)}\\
&\alpha^2 Q(P+N_e)-2PQ\alpha-P^2>0.
\end{align} 
Therefore, $\alpha\ge\frac{P}{P+N_e}\left(1+\sqrt\frac{P+Q+N_e}{N_e}\right)=\alpha_0$
or $\alpha\le\frac{P}{P+N_e}\left(1-\sqrt\frac{P+Q+N_e}{N_e}\right)=\alpha_{-0}$. Similarly, for inequality $I(U;S)< I(U;Y_e)$ we have: $\frac{P}{P+N_e}\left(1-\sqrt\frac{P+Q+N_e}{N_e}\right)=\alpha_{-0}<\alpha<\frac{P}{P+N_e}\left(1+\sqrt\frac{P+Q+N_e}{N_e}\right)=\alpha_0$.
\section{Proof of Theorem \ref{thm-sc}}\label{app-thm-sc}
Here, exploiting the same coding method used in \cite{chen2008wiretap,chia2012wiretap,liu2007wiretap} to achieve the SC, we prove Theorem \ref{thm-sc} considering the known fading coefficients $h_m$ and $h_e$. To this end, we assume $U_m=X+\alpha_mS$ and $U_e=X+\alpha_eS$ as used for generalized dirty paper coding \cite{costa1983writing}, where $X$ and $S$ are independent RVs distributed according to $\mathcal{N}\left(0,P\right)$ and $\mathcal{N}\left(0,Q\right)$, respectively, and $\alpha_m$ and $\alpha_e$ are parameters to be later determined. So, by considering \eqref{eq-ym}, we compute the values of  mutual information $I(U_m;Y_m)$ and $I(U_m;S)$ as follows:
\begin{align}
	&I(U_m;Y_m)=I(X+\alpha_mS;h_mX+S+Z_m)\\
	&=H(h_mX+S+Z_m)-H(h_mX+S+Z_m|X+\alpha_mS)\\\nonumber
	&=H(h_mX+S+Z_m)+H(X+\alpha_mS)\\
	&\quad-H(h_mX+S+Z_m,X+\alpha_mS)\\\nonumber
	&=\log\left(\left(2\pi\mathrm{e}\right)^2\left(|h_m|^2P+Q+N_m\right)\left(P+\alpha_m^2Q\right)\right)\\
	&\quad-\log\left(\left(2\pi\mathrm{e}\right)^2\det\left(\mathrm{cov}\left(h_mX+S+Z_m,X+\alpha_mS\right)\right)\right)\\
	&=\log\left(\frac{\left(|h_m|^2P+Q+N_m\right)\left(P+\alpha_m^2Q\right)}{PQ\left(1-|h_m|\alpha_m\right)^2+N_m\left(P+\alpha_m^2Q\right)}\right)\label{eq-Iy},
\end{align}
and
\begin{align}
	&I(U_m;S)=I(X+\alpha_mS;S)\\
	&=H(X+\alpha_mS)+H(S)-H(X+\alpha_mS,S)\\\nonumber
	&=\log\left(\left(2\pi\mathrm{e}\right)^2\left(P+\alpha_m^2Q\right)Q\right)\\
	&\quad-\log\left(\left(2\pi\mathrm{e}\right)^2\left(\left(P+\alpha_m^2Q\right)Q-\alpha_m^2Q^2\right)\right)\\
	&=\log\left(\frac{P+\alpha_m^2Q}{P}\right)\label{eq-Is}.
\end{align}
Now, by inserting \eqref{eq-Iy} and \eqref{eq-Is} into $\mathcal{C}_s$ defined in Corollary \ref{col-cm}, we have:
\begin{align}
	\mathcal{C}_m(\alpha_m)=\log\left(\frac{P\left(|h_m|^2P+Q+N_m\right)}{PQ\left(1-|h_m|\alpha_m\right)^2+N_m\left(P+\alpha_m^2Q\right)}\right).
\end{align}
Consequently, by maximizing $\mathcal{C}_m(\alpha_m)$ over $\alpha_m$, the SC $\mathcal{C}_s$ can be obtained as:
\begin{align}
	\mathcal{C}_s&=\underset{\alpha_m}{\max}\;\mathcal{C}_m(\alpha_m)\\
	&=\log\left(1+\frac{|h_m|^2P}{N_m}\right)\\
	&=\log\left(1+\gamma_m\right).
\end{align}
In order to prove the SC $\mathcal{C}_s$ for the second condition, we need to calculate the values of  mutual information $I(U_e;Y_m)$ and $I(U_e;Y_e)$. Similarly, by considering $\eqref{eq-ye}$ and $U_e=X+\alpha_eS$, we have:
\begin{align}\label{eq-Iym}
	I(U_e;Y_m)=\log\left(\frac{\left(|h_m|^2P+Q+N_m\right)\left(P+\alpha_e^2Q\right)}{PQ\left(1-|h_m|\alpha_e\right)^2+N_m\left(P+\alpha_e^2Q\right)}\right),
\end{align}
and
\begin{align}\label{eq-Iye}
	I(U_e;Y_e)=\log\left(\frac{\left(|h_e|^2P+Q+N_e\right)\left(P+\alpha_e^2Q\right)}{PQ\left(1-|h_e|\alpha_e\right)^2+N_e\left(P+\alpha_e^2Q\right)}\right).
\end{align}
Now, by substituting \eqref{eq-Iym} and \eqref{eq-Iye} in $\mathcal{C}_s$ defined in Corollary \ref{col-cs}, we obtain:
{\small\begin{align}\nonumber
	&R_e(\alpha_e)=\\
	&\log\hspace{-0.1cm}\left(\frac{\left(|h_m|^2P+Q+N_m\right)\left(PQ\left(1-|h_e|\alpha_e\right)^2+N_e\left(P+\alpha_e^2Q\right)\right)}{\left(|h_e|^2P+Q+N_e\right)\left(PQ\left(1-|h_m|\alpha_e\right)^2+N_m\left(P+\alpha_e^2Q\right)\right)}\right).
\end{align}}
Finally, by maximizing $R_e(\alpha_e)$ over $\alpha_e$, the SC $\mathcal{C}_s$ can be determined as:
\begin{align}
	\mathcal{C}_s=\underset{\alpha_e}{\max}\;R_e(\alpha_e)=\log\left(\frac{1+\bar{\gamma}_{ms}+\gamma_m}{1+\bar{\gamma}_{es}+\gamma_e}\right).
\end{align}
\section{Proof of Theorem \ref{thm-asc}}\label{app-thm-asc}
\textit{Proof of} $\bar{\mathcal{C}}_s^1$: By substituting the joint PDF $f(\gamma_m,\gamma_{e})$ in to \eqref{eq-sc1}, the ASC under condition of the Corollary \ref{col-cm} can be determined as:
\begin{align}\nonumber
\bar{\mathcal{C}}_s^1&=\frac{1}{\bar{\gamma}_m\bar{\gamma}_e}\int_{0}^{\infty}\int_{0}^{\bar{\gamma}}\log_2\left(1+\gamma_m\right)\mathrm{e}^{-\frac{\gamma_m}{\bar{\gamma}_m}-\frac{\gamma_e}{\bar{\gamma}_e}}\\
&\quad\times\Big[1+\theta\big(1-2e^{-\frac{\gamma_m}{\bar{\gamma}_m}}\big)\big(1-2e^{-\frac{\gamma_e}{\bar{\gamma}_e}}\big)\Big]d\gamma_{e}d\gamma_m\\\nonumber
&=\frac{1}{\bar{\gamma}_{m}\bar{\gamma}_e}\int_{0}^{\infty}\int_{0}^{\bar{\gamma}}\mathrm{e}^{-\frac{\gamma_m}{\bar{\gamma}_m}-\frac{\gamma_e}{\bar{\gamma}_e}}\log_2\left(1+\gamma_{m}\right)d\gamma_{e}d\gamma_m\\\nonumber
&+\theta\Bigg[\frac{1}{\bar{\gamma}_{m}\bar{\gamma}_e}\int_{0}^{\infty}\int_{0}^{\bar{\gamma}}\mathrm{e}^{-\frac{\gamma_m}{\bar{\gamma}_m}-\frac{\gamma_e}{\bar{\gamma}_e}}\log_2\left(1+\gamma_{m}\right)d\gamma_{e}d\gamma_m\\\nonumber
&\quad-\frac{2}{\bar{\gamma}_{m}\bar{\gamma}_e}\int_{0}^{\infty}\int_{0}^{\bar{\gamma}}\mathrm{e}^{-\frac{2\gamma_m}{\bar{\gamma}_m}-\frac{\gamma_e}{\bar{\gamma}_e}}\log_2\left(1+\gamma_{m}\right)d\gamma_{e}d\gamma_m\\\nonumber
&\quad-\frac{2}{\bar{\gamma}_{m}\bar{\gamma}_e}\int_{0}^{\infty}\int_{0}^{\bar{\gamma}}\mathrm{e}^{-\frac{\gamma_m}{\bar{\gamma}_m}-\frac{2\gamma_e}{\bar{\gamma}_e}}\log_2\left(1+\gamma_{m}\right)d\gamma_{e}d\gamma_m\\
&\quad+\frac{4}{\bar{\gamma}_{m}\bar{\gamma}_e}\int_{0}^{\infty}\int_{0}^{\bar{\gamma}}\mathrm{e}^{-\frac{2\gamma_m}{\bar{\gamma}_m}-\frac{2\gamma_e}{\bar{\gamma}_e}}\log_2\left(1+\gamma_{m}\right)d\gamma_{e}d\gamma_m\Bigg]\\
&=\mathcal{W}_1+\theta\left[\mathcal{W}_1-2\mathcal{W}_2-2\mathcal{W}_3+4\mathcal{W}_4\right],
\end{align}
where the integrals $\mathcal{W}_{\eta}$, for $\eta\in\{1,2,3,4\}$, are in the following format \cite{zwillinger2007table}:
\begin{align}
	&\int_{0}^{\infty}\mathrm{e}^{-\zeta t}\ln\left(1+t\right)dt=\frac{\mathrm{e}^{\zeta}}{\zeta}\mathrm{E}_1\left(\zeta\right)\label{int-11}.
\end{align}
Thus, by exploiting \eqref{int-11}, $\mathcal{W}_{\eta}$ can be computed, and then the proof is completed. 

\textit{Proof of} $\bar{\mathcal{C}}_s^2$: By applying Lemma \ref{lemma-pdf} to ASC definition in \eqref{eq-sc2} and exploiting the linearity rules of integration, \eqref{eq-sc2} can be decomposed as:
\begin{align}\nonumber
	\bar{\mathcal{C}}&=\frac{1}{\bar{\gamma}_m\bar{\gamma}_e}\int_{0}^{\infty}\int_{0}^{\bar{\gamma}}\log_2\left(1+\bar{\gamma}_{ms}+\gamma_m\right)\mathrm{e}^{-\frac{\gamma_m}{\bar{\gamma}_m}-\frac{\gamma_e}{\bar{\gamma}_e}}\\\nonumber
	&\times\Big[1+\theta\big(1-2e^{-\frac{\gamma_m}{\bar{\gamma}_m}}\big)\big(1-2e^{-\frac{\gamma_e}{\bar{\gamma}_e}}\big)\Big]d\gamma_{e}d\gamma_m\\\nonumber
	&-\frac{1}{\bar{\gamma}_m\bar{\gamma}_e}\int_{0}^{\infty}\int_{0}^{\bar{\gamma}}\log_2\left(1+\bar{\gamma}_{es}+\gamma_e\right)\mathrm{e}^{-\frac{\gamma_m}{\bar{\gamma}_m}-\frac{\gamma_e}{\bar{\gamma}_e}}\\\nonumber
	&\times\Big[1+\theta\big(1-2e^{-\frac{\gamma_m}{\bar{\gamma}_m}}\big)\big(1-2e^{-\frac{\gamma_e}{\bar{\gamma}_e}}\big)\Big]d\gamma_{e}d\gamma_m\\
	&=\mathcal{K}-\mathcal{M},
\end{align}
where $\bar{\gamma}=\gamma_m+\bar{\gamma}_{ms}-\bar{\gamma}_{es}$, 
\begin{align}\nonumber
&\mathcal{K}=\frac{1}{\bar{\gamma}_m\bar{\gamma}_e}\int_{0}^{\infty}\int_{0}^{\bar{\gamma}}\log_2\left(1+\bar{\gamma}_{ms}+\gamma_m\right)\mathrm{e}^{-\frac{\gamma_m}{\bar{\gamma}_m}-\frac{\gamma_e}{\bar{\gamma}_e}}d\gamma_{e}d\gamma_m\\\nonumber
&+\theta\Bigg[\frac{1}{\bar{\gamma}_m\bar{\gamma}_e}\int_{0}^{\infty}\int_{0}^{\bar{\gamma}}\log_2\left(1+\bar{\gamma}_{ms}+\gamma_m\right)\mathrm{e}^{-\frac{\gamma_m}{\bar{\gamma}_m}-\frac{\gamma_e}{\bar{\gamma}_e}}d\gamma_{e}d\gamma_m\\\nonumber
&-\frac{2}{\bar{\gamma}_m\bar{\gamma}_e}\int_{0}^{\infty}\int_{0}^{\bar{\gamma}}\log_2\left(1+\bar{\gamma}_{ms}+\gamma_m\right)\mathrm{e}^{-\frac{2\gamma_m}{\bar{\gamma}_m}-\frac{\gamma_e}{\bar{\gamma}_e}}d\gamma_{e}d\gamma_m\\\nonumber
&-\frac{2}{\bar{\gamma}_m\bar{\gamma}_e}\int_{0}^{\infty}\int_{0}^{\bar{\gamma}}\log_2\left(1+\bar{\gamma}_{ms}+\gamma_m\right)\mathrm{e}^{-\frac{\gamma_m}{\bar{\gamma}_m}-\frac{2\gamma_e}{\bar{\gamma}_e}}d\gamma_{e}d\gamma_m\\
&+\frac{4}{\bar{\gamma}_m\bar{\gamma}_e}\int_{0}^{\infty}\int_{0}^{\bar{\gamma}}\log_2\left(1+\bar{\gamma}_{ms}+\gamma_m\right)\mathrm{e}^{-\frac{2\gamma_m}{\bar{\gamma}_m}-\frac{2\gamma_e}{\bar{\gamma}_e}}d\gamma_{e}d\gamma_m\Bigg]\\
&=\mathcal{K}_1+\theta\left[\mathcal{K}_1-2\mathcal{K}_2-2\mathcal{K}_3+4\mathcal{K}_4\right],
\end{align}
and
\begin{align}\nonumber
&\mathcal{M}=\frac{1}{\bar{\gamma}_m\bar{\gamma}_e}\int_{0}^{\infty}\int_{0}^{\bar{\gamma}}\log_2\left(1+\bar{\gamma}_{es}+\gamma_e\right)\mathrm{e}^{-\frac{\gamma_m}{\bar{\gamma}_m}-\frac{\gamma_e}{\bar{\gamma}_e}}d\gamma_{e}d\gamma_m\\\nonumber
&+\theta\Bigg[\frac{1}{\bar{\gamma}_m\bar{\gamma}_e}\int_{0}^{\infty}\int_{0}^{\bar{\gamma}}\log_2\left(1+\bar{\gamma}_{es}+\gamma_e\right)\mathrm{e}^{-\frac{\gamma_m}{\bar{\gamma}_m}-\frac{\gamma_e}{\bar{\gamma}_e}}d\gamma_{e}d\gamma_m\\\nonumber
&-\frac{2}{\bar{\gamma}_m\bar{\gamma}_e}\int_{0}^{\infty}\int_{0}^{\bar{\gamma}}\log_2\left(1+\bar{\gamma}_{es}+\gamma_e\right)\mathrm{e}^{-\frac{2\gamma_m}{\bar{\gamma}_m}-\frac{\gamma_e}{\bar{\gamma}_e}}d\gamma_{e}d\gamma_m\\\nonumber
&-\frac{2}{\bar{\gamma}_m\bar{\gamma}_e}\int_{0}^{\infty}\int_{0}^{\bar{\gamma}}\log_2\left(1+\bar{\gamma}_{es}+\gamma_e\right)\mathrm{e}^{-\frac{\gamma_m}{\bar{\gamma}_m}-\frac{2\gamma_e}{\bar{\gamma}_e}}d\gamma_{e}d\gamma_m\\
&+\frac{4}{\bar{\gamma}_m\bar{\gamma}_e}\int_{0}^{\infty}\int_{0}^{\bar{\gamma}}\log_2\left(1+\bar{\gamma}_{es}+\gamma_e\right)\mathrm{e}^{-\frac{2\gamma_m}{\bar{\gamma}_m}-\frac{2\gamma_e}{\bar{\gamma}_e}}d\gamma_{e}d\gamma_m\Bigg]\\
&=\mathcal{M}_1+\theta\left[\mathcal{M}_1-2\mathcal{M}_2-2\mathcal{M}_3+4\mathcal{M}_4\right].
\end{align}\
The integrals $\mathcal{K}_{\eta}$ and $\mathcal{M}_{\eta}$, for $\eta\in\{1,2,3,4\}$, are in the following formats \cite{zwillinger2007table}:
\begin{align}\nonumber
&\int_{}^{}\mathrm{e}^{-\zeta t}\ln\left(1+\kappa+t\right)dt\\
&=-\frac{1}{\zeta}\left[\mathrm{e}^{\zeta(1+\kappa)}\mathrm{E}_1\left(\zeta\left(1+\kappa+t\right)\right)+\mathrm{e}^{-\zeta t}\ln\left(1+\kappa+t\right)\right]\label{int-1},
\end{align}
\begin{align}\nonumber
&\int_{0}^{\infty}\mathrm{e}^{-\zeta t}\ln\left(1+\kappa+t\right)dt\\
&=\frac{1}{\zeta}\left[\mathrm{e}^{\zeta(1+\kappa)}\mathrm{E}_1\left(\zeta\left(1+\kappa\right)\right)+\ln\left(1+\kappa\right)\right]\label{int-2},
\end{align}
\begin{align}
	&\int_{0}^{\infty}\mathrm{e}^{-\zeta t}\mathrm{E}_1\left(\delta+\nu t\right)=\frac{1}{\zeta}\left[\mathrm{E}_1\left(\delta\right)-\mathrm{e}^{\frac{\zeta\delta}{\nu}}\mathrm{E}_1\left(\frac{\delta(\zeta+\nu)}{\nu}\right)\right]\label{int-4}.
\end{align}
Therefore, by utilizing the integral formats in \eqref{int-1}--\eqref{int-4} and conducting the required simplifications, $\mathcal{K}$, $\mathcal{M}$, and then $\bar{\mathcal{C}}_s^{2}$ are computed. 
\section{Proof of Theorem \ref{thm-sop}}\label{app-thm-sop}
First, we prove the SOP under condition of the Corollary \ref{col-cm}. By applying Lemma \ref{lemma-pdf} to \eqref{eq-sop-1}, $P_{sop}^1$ can be rewritten as: 
\begin{align}\nonumber
	P_{sop}^{1}&=1-\int_{0}^{\infty}\int_{2^{R_s}-1}^{\infty}\frac{e^{-\frac{\gamma_m}{\bar{\gamma}_m}-\frac{\gamma_e}{\bar{\gamma}_e}}}{\bar{\gamma}_m\bar{\gamma}_e}\\
	&\quad\times\left[1+\theta\left(1-2e^{-\frac{\gamma_m}{\bar{\gamma}_m}}\right)\left(1-2e^{-\frac{\gamma_e}{\bar{\gamma}_e}}\right)\right]d\gamma_m d\gamma_e\\\nonumber
	&=1-\Bigg[\int_{0}^{\infty}\int_{2^{R_s}-1}^{\infty}\frac{\mathrm{e}^{-\frac{\gamma_m}{\bar{\gamma}_m}-\frac{\gamma_e}{\bar{\gamma}_e}}}{\bar{\gamma}_m\bar{\gamma}_e}d\gamma_m d\gamma_e\\\nonumber
	&\quad+\theta\Bigg(\int_{0}^{\infty}\int_{2^{R_s}-1}^{\infty}\frac{\mathrm{e}^{-\frac{\gamma_m}{\bar{\gamma}_m}-\frac{\gamma_e}{\bar{\gamma}_e}}}{\bar{\gamma}_m\bar{\gamma}_e}d\gamma_m d\gamma_e\\\nonumber
	&\quad-2\int_{0}^{\infty}\int_{2^{R_s}-1}^{\infty}\frac{\mathrm{e}^{-\frac{2\gamma_m}{\bar{\gamma}_m}-\frac{\gamma_e}{\bar{\gamma}_e}}}{\bar{\gamma}_m\bar{\gamma}_e}d\gamma_m d\gamma_e\\\nonumber
	&\quad-2\int_{0}^{\infty}\int_{2^{R_s}-1}^{\infty}\frac{\mathrm{e}^{-\frac{\gamma_m}{\bar{\gamma}_m}-\frac{2\gamma_e}{\bar{\gamma}_e}}}{\bar{\gamma}_m\bar{\gamma}_e}d\gamma_m d\gamma_e\\
	&\quad+4\int_{0}^{\infty}\int_{2^{R_s}-1}^{\infty}\frac{\mathrm{e}^{-\frac{2\gamma_m}{\bar{\gamma}_m}-\frac{2\gamma_e}{\bar{\gamma}_e}}}{\bar{\gamma}_m\bar{\gamma}_e}d\gamma_m d\gamma_e\Bigg)\Bigg]\\
	&=1-\left[\mathcal{I}_1+\theta\left(\mathcal{I}_1-2\mathcal{I}_2-2\mathcal{I}_3+4\mathcal{I}_4\right)\right],
\end{align}
where the integrals $\mathcal{I}_{\eta}$, for $\eta\in\{1,2,3,4\}$ can be solved easily by performing simple calculations, and then $P_{sop}^1$ is obtained as \eqref{eq-sop1}.

Similarly, by inserting the joint PDF $f(\gamma_m,\gamma_{e})$ into \eqref{eq-sop2}, $P_{sop}^2$ can be obtained as:
\begin{align}\nonumber
	P_{sop}^{2}&=1-\int_{0}^{\infty}\int_{\gamma_{th}}^{\infty}\frac{e^{-\frac{\gamma_m}{\bar{\gamma}_m}-\frac{\gamma_e}{\bar{\gamma}_e}}}{\bar{\gamma}_m\bar{\gamma}_e}\\
	&\quad\times\left[1+\theta\left(1-2e^{-\frac{\gamma_m}{\bar{\gamma}_m}}\right)\left(1-2e^{-\frac{\gamma_e}{\bar{\gamma}_e}}\right)\right]d\gamma_m d\gamma_e\\\nonumber
	&=1-\Bigg[\int_{0}^{\infty}\int_{\gamma_{th}}^{\infty}\frac{\mathrm{e}^{-\frac{\gamma_m}{\bar{\gamma}_m}-\frac{\gamma_e}{\bar{\gamma}_e}}}{\bar{\gamma}_m\bar{\gamma}_e}d\gamma_m d\gamma_e\\\nonumber
	&\quad+\theta\Bigg(\int_{0}^{\infty}\int_{\gamma_{th}}^{\infty}\frac{\mathrm{e}^{-\frac{\gamma_m}{\bar{\gamma}_m}-\frac{\gamma_e}{\bar{\gamma}_e}}}{\bar{\gamma}_m\bar{\gamma}_e}d\gamma_m d\gamma_e\\\nonumber
	&\quad-2\int_{0}^{\infty}\int_{\gamma_{th}}^{\infty}\frac{\mathrm{e}^{-\frac{2\gamma_m}{\bar{\gamma}_m}-\frac{\gamma_e}{\bar{\gamma}_e}}}{\bar{\gamma}_m\bar{\gamma}_e}d\gamma_m d\gamma_e\\\nonumber
	&\quad-2\int_{0}^{\infty}\int_{\gamma_{th}}^{\infty}\frac{\mathrm{e}^{-\frac{\gamma_m}{\bar{\gamma}_m}-\frac{2\gamma_e}{\bar{\gamma}_e}}}{\bar{\gamma}_m\bar{\gamma}_e}d\gamma_m d\gamma_e\\
	&\quad+4\int_{0}^{\infty}\int_{\gamma_{th}}^{\infty}\frac{\mathrm{e}^{-\frac{2\gamma_m}{\bar{\gamma}_m}-\frac{2\gamma_e}{\bar{\gamma}_e}}}{\bar{\gamma}_m\bar{\gamma}_e}d\gamma_m d\gamma_e\Bigg)\Bigg]\\
	&=1-\left[\mathcal{J}_1+\theta\left(\mathcal{J}_1-2\mathcal{J}_2-2\mathcal{J}_3+4\mathcal{J}_4\right)\right],
\end{align}
where $\bar{\gamma}_{th}=\left(2^{R_s}\left(1+\bar{\gamma}_{es}\right)-\left(1+\bar{\gamma}_{ms}\right)\right)$. The integrals $\mathcal{J}_{\eta}$, for ${\eta}\in\{1,2,3,4\}$, can be computed by some simple calculations. Then, after conducting the required simplifications, the proof is completed. 

\bibliographystyle{IEEEtran}
\bibliography{sample.bib}

\end{document}